\documentclass[preprint,12pt]{elsarticle}

\usepackage{graphicx}
\usepackage{amssymb}
\usepackage{amsmath}
\usepackage{amsthm}
\usepackage{float}
\usepackage[pdftex]{hyperref}
\usepackage{natbib}
\usepackage[prependcaption, textsize=tiny, linecolor=red, backgroundcolor=red!25, bordercolor=red, textwidth=3.5cm]{todonotes}

\usepackage{lineno}

\renewcommand{\citeauthoryear}[1]{\citeauthor{#1} (\citeyear{#1})}

\newcommand{\walk}{w}		
\newcommand{\walkL}{l}		
\newcommand{\walkS}{s}	
\newcommand{\ppath}{p}		
\newcommand{\ppathL}{l^\prime}		
\newcommand{\ppathS}{s^\prime}	

\newtheorem{theorem}{Theorem}[section]

\newtheorem{definition}{Definition}[section]

\newtheorem{remark}{Remark}[section]

\journal{Journal Name}

\begin{document}

\begin{frontmatter}


\title{Beyond the shortest path: the path length index as a distribution}



\author[Cemaden]{Leonardo B. L. Santos}\ead{santoslbl@gmail.com}

\author[FGV]{Luiz Max Carvalho}

\author[INPE]{Giovanni G. Soares}

\author[MP]{Leonardo N. Ferreira}

\author[HUB]{Igor M. Sokolov}

\address[Cemaden]{National Center for Monitoring and Early Warning of Natural Disasters (Cemaden), Brazil}  

\address[FGV]{School of Applied Mathematics (EMAp), Getulio Vargas Foundation (FGV), Brazil}

\address[INPE]{National Institute of Space Research (INPE), Brazil}  

\address[MP]{Center for Humans and Machines, Max Planck Institute for Human Development, Germany}

\address[HUB]{Humboldt University of Berlin, Germany}

\end{frontmatter}



\section{Motivation}
\label{sec:introduction}

Traversing graphs is a fundamental question in Graph Theory and an important consideration when dealing with complex networks.
The traditional complex network approach considers only the shortest paths from one node to another \cite{barabasi2016network}, and does not take into account several other possible paths.
This limitation is significant, for example, in urban mobility studies \cite{Lima2016, Galbrun2016}, where it important to consider alternative routes between locations.

As mentioned by \citeauthoryear{Lima2016}, in urban mobility settings users choose multiple routes over origin-destination pairs, and those choices often deviate from the shortest time path. \citeauthoryear{Galbrun2016} highlight that chosen routes may be associated with diverse factors, for instance public safety. \citeauthoryear{Tomas2022} further support these claims by showing that exceptional events, such as urban floods, may lead users to deviate from routes previously defined to risk-less (and potentially longer) ones.

\citeauthoryear{Estrada2008} proposed the Communicability Index, a number (scalar) that takes into account not only the shortest paths but also all the walks from one node to another.
Their approach was based on walks.
In contrast, here we are interested in paths due to the urban mobility motivation context. 

On one hand, the number of walks between each pair of nodes in a simple graph is known analytically~\cite{Biggs}.
On the other hand, the analogous problem for paths is NP-hard \cite{Roberts2007}.
\citeauthoryear{Roberts2007} presented an stochastic algorithm to estimate the solution of that problem using a sequential importance sampling.

In this short report, as the first steps, we present an exhaustive approach to address the problem of finding all paths between two nodes.
We show one can go beyond the shortest path but we do not need to go so far: we present an interactive procedure and an early stop possibility.
We apply our ideas to the well-known Zachary's karate club graph~\cite{Zachary1977}.
We do not collapse the distribution of path lengths between a pair of nodes into a scalar number; instead we look at the distribution itself - taking all paths up to a pre-defined path length (considering a truncated distribution), and show the impact of that approach on the most straightforward distance-based graph index: the walk/path length.


\section{Preliminaries: definitions and notation}
\label{sec:notation}

In this section, we give a few definitions and results from elementary graph theory to facilitate the understanding of the new results presented herein. Most of the following discussion is standard and can be found in~\cite{barabasi2016network, Biggs}.

We start by defining a graph (Definition~\ref{def:Graph}) and then a simple graph (Definition~\ref{def:SimpleGraph}), which will be the main objects of interest in this paper. 

\begin{definition}[\textbf{Graph}]
\label{def:Graph}
A graph $G = (V, E)$ is a set of nodes and edges, where $V$ is the set of $|V|$ = $N$ nodes and $E$ is the set of $|E|$ = $M$ edges.
\end{definition}

An edge (also called link or connection) $(i,j), i,j\in V$ connects two nodes $i$ and $j$. A self-connection or a loop is a link $(i,i)$ that connects node $i$ to itself. Multiple edges are two or more edges that connect the same two vertices. 

A link can be undirected or directed.
In an undirected graph, all edges $(i,j)$ connect $i$ to $j$ and vice-versa. A directed graph has directed edges (also called arcs) $(i,j)$, that connect $i$ to $j$, but not $j$ to $i$, i.e., $(i,j) \neq (j,i)$.
A link can also have an associated weight, which is a numeric value.

\begin{definition}[\textbf{Simple Graph}]
\label{def:SimpleGraph}
A graph $G = (V, E)$ is a simple graph if, and only if, it is undirected, there are no self-connections in $G$, no multiple edges or weights.
\end{definition}

A helpful object for characterizing a graph is its adjacency matrix, whose definition is given in Definition~\ref{def:adjacency_matrix}.
\begin{definition}[\textbf{Adjacency matrix}]
\label{def:adjacency_matrix}
The adjacency matrix $\boldsymbol{A}$ of a graph $G = (V, E)$ is the $N \times N$ matrix whose entries $A_{ij}$ are given by
\begin{equation*}
    A_{ij} = \begin{cases}
    1,\: \text{if} \quad i, j \in V \: \text{share an edge};\\
    0,\: \text{otherwise}.
    \end{cases}
\end{equation*}

\end{definition}

In this paper, we are concerned with traversing the graph, i.e., starting from a source node $i \in V$, visiting a collection of nodes, and arrive a target node $j \in V$, where $i = j$ is a possibility.
Here we distinguish between trajectories that allow multiple visits to the same node (and associated) edges, called~\textbf{walks} (Definition~\ref{def:Walk}); and trajectories where each node and vertex can only be visited once, called~\textbf{paths}, presented in Definition~\ref{def:Path}.

We start our discussion with trajectories that can visit the same node multiple times, called~\textbf{walks}:
\begin{definition}[\textbf{Walk}]
\label{def:Walk}
Consider a simple graph $G = (V, E)$ and a pair of nodes $i, j$ in $V$.
A walk $\walk$ in $G$ from $i$ to $j$ is an alternating sequence of edges and nodes from $i$ (node of origin/source) to $j$ (node of destination/target).
\end{definition}
With this definition in hand, we are prepared to state Theorem~\ref{thm:finite_walks}, which tells us that the number of walks of a given (finite) length is finite so long as $|V|$ is finite.
\begin{theorem}[\textbf{Finite number of walks}]
\label{thm:finite_walks}
Consider a simple graph $G = (V, E)$.
Take $i, j \in V$, the number of walks of length $\walkL$ between $i$ and $j$ is given by
\begin{equation*}
    f_{W}^{ij}(\walkL) = \left(A^{\walkL}\right)_{ij},
\end{equation*}
where $A_{ij}$ is the corresponding entry in the adjacency matrix of $G$ -- see Definition~\ref{def:adjacency_matrix}.
\end{theorem}
\begin{proof}
This is a well-known result.
See Lemma 2.5 in~\cite{Biggs}.
\end{proof}

Now, consider trajectories in a graph without ever visiting any node twice.
Such a trajectory is called a~\textbf{path}:
\begin{definition}[\textbf{Path}]
\label{def:Path}
Consider a simple graph $G = (V, E)$ and a pair of nodes $i, j$ in $|V|$.
A path $\ppath_{ij}$ in $G$ from $i$ to $j$ is an open ($i \ne j$) walk from $i$ to $j$, and with no repeated edges or nodes.
\end{definition}
As Definition~\ref{def:Path} makes clear, paths are specializations (restrictions) of walks.
This might prompt the reader to think that one can study paths by considering restrictions to results about walks.
As we will show later on, this is not always the case.
Our approach is somehow similar to Self Avoiding Walks (SAW) \cite{SAW}, but we fix not only the source but also the target for each path.

In this paper, we will devote attention to connected graphs (Definition~\ref{def:ConnectedGraph}), that is, graphs for which there exists at least one path for every pair of vertices $i,j \in V$.
\begin{definition}[\textbf{Connected Graph}]
\label{def:ConnectedGraph}
A simple graph $G = (V, E)$ is connected if for every pair of vertices $i, j$ one can construct a subset $C_{ij} \subseteq V$, with $|C_{ij}| = K$ where the vertices $c_1, \ldots, c_K \in C_{ij}$ are such that $i$ and $c_1$ share and edge as do $j$ and $c_K$ and also $c_k$ and $c_{k+1}$ share an edge, for $ 2 \leq k \leq K-1$.
In other words, $G$ is connected if and only if one can always construct at least one path between $i, j \in V$, for every such pair. 
\end{definition}


The number of vertices visited in the path $\ppath_{ij}$ is the path length (Definition~\ref{def:PathLength}), and the shortest such path (Definition~\ref{def:SPL}) is usually of great interest as it is related to many optimization problems, such as the traveling salesman problem.
We now state a few more definitions related to traversal of graphs, which will be useful in the remainder of the paper.

\begin{definition}[\textbf{Path Length}]
\label{def:PathLength}

Consider a simple graph $G = (V, E)$.
The number of edges on the path from $i$ to $j$ is the path length ($\ppathL$) of that path.
\end{definition}

\begin{definition}[\textbf{Shortest Path Length}]
\label{def:SPL}
Consider a simple graph $G = (V, E)$.
The number of edges on the shortest path from $i$ to $j$ is the shortest path length ($\ppathS$) of that path. The $\ppathS$ is a number associated with the pair $i$-$j$: for each pair $i$-$j$ there is one and only one $\ppathS$: $\ppathS$(i,j).
\end{definition}

\begin{remark}[\textbf{Shortest Path}]
\label{rmk:ShortestPath}
Consider a simple graph $G = (V, E)$.
For any two vertices $i, j \in V$ there is at least one path from $i$ to $j$ which the path length is the shortest possible.
\end{remark}

\section{Main problem: computing the frequency and length of walks and paths}
\label{sec:main}


First, let us take a look at the number of walks: in a simple graph $G = (V, E)$, for any two vertices $i, j \in V$, there is an infinite number of walks from $i$ to $j$, which holds even for finite graphs.
However, if we take a finite path length, the number of walks with that path length is finite, and it is given by $f(\walkL_{ij}) = (A^n)_{ij}$, with $n = \walkL_{ij}$.
One might now ask what the expected value for $\walkL_{ij}$ is.
As we can always get a walk longer than any other, we cannot define a normalized probability measure; thus, this expectation does not exist.

Let us define the shortest walk length from $i$ to $j$, $\walkS_{ij}$, as the minimum value of $\walkL_{ij}$ - obviously that the walk associated with this length is a path. Then, once any (finite) $\walkL_{ij}$ can be expressed as $\walkL_{ij}$ = $\walkS_{ij}$ + $k$, $k \ge 0$, $k \in N$, and, therefore, a ``truncated'' expected value, under a k-th order approximation, is: 
	\begin{equation}
    \mathbb{E}[\walkL] = \frac
{\sum_{n=\walkS_{ij}}^{n=\walkS_{ij}+k} n (A^n)_{ij}}
{\sum_{n=\walkS_{ij}}^{n=\walkS_{ij}+k} (A^n)_{ij}}.
    \end{equation}

Now, let us move from walks to paths.
Between any pair of nodes $i$-$j$ in $G$, there is at least one path $\ppath$, from $i$ to $j$ - we are considering a single connected component in $G$. 
While \# $\ppath$ is finite, the problem of counting the number of s-t (source-destination) paths in a graph is NP-complete \cite{Roberts2007}.

Here we propose a $k$-th order approximation for the case of paths.
The length of $\ppath_{ij}$ is $\ppathL_{ij}$, and it is between $1$ and $N-1$.
Let us define the shortest path length from $i$ to $j$, s$\ppathL$, as the minimum value of $\ppathL_{ij}$.
Any $\ppathL_{ij}$, therefore, can be expressed as $\ppathL_{ij} = \ppathS_{ij} + k$, for a finite value of $k \le N-2$.
Finally, the expected value, under the k-th order approximation, is:
	\begin{equation}
		\mathbb{E}[\ppathL] = \frac
{\sum_{n=\ppathS_{ij}}^{n=\ppathS_{ij}+k} n f(n)}
{\sum_{n=\ppathS_{ij}}^{n=\ppathS_{ij}+k} f(n)},
	\end{equation}
	 where $f(n)$ is the frequency of a $\ppathL_{ij}$=n.

There is no analytical expression for $f(n)$ in the literature.
Finding all paths in a graph can be very computationally expensive - $O(N!)$ in the worst case: a complete graph with order $N$.
Here we perform a depth-limited search (DLS) in order to find $f(n)$, which can found \href{https://github.com/gioguarnieri/all_paths}{here.}\footnote{\url{https://github.com/gioguarnieri/all_paths}}\footnote{We discussed ``to go beyond the shortest path'' in 2018 and implemented the first complete version of this code in August 2019. The COVID-19 pandemic has changed research agendas worldwide. We resume this paper in 2022.}.

It is worth highlighting we do an exhaustive search - finding all possible paths from a node to another.
However, the main insight is that we do not need to go so far beyond the shortest paths - in order words: we do not use a so much high value of $k$ in the $k$-th order approximation.

\section{Results}
\label{sec:results}

In this section, we present an analytical result considering complete graphs (\ref{sec:complete_graphs}), and, based on a depth-limited Search, results for the Zachary's Karate Club graph (\ref{sec:zachary}).

\subsection{Complete graphs}
\label{sec:complete_graphs}

In a complete graph all nodes are directly connected to all others ($\ppathS_{ij}$=1, $\forall$ i,j).
The number of paths between any pair of nodes is a combinatorial result based on the arrange of N-2 nodes in a path of length $\ppathL$.
\begin{theorem}[\textbf{Number of paths in a complete graph}]
Let $G$ be a complete simple graph.
Then the number of paths of length $k + 1$ is
\begin{equation*}
    f(k + 1) = 
    \begin{cases}
    1, \quad k = 0, \\
    \prod_{r=2}^{r=k+1} N - r, \quad 0 < k < N-1,\\
    0, \quad  k \geq N-1
    \end{cases}
\end{equation*}
\end{theorem}

It is possible to note that, in a simple but complete graph, between any pair of nodes:

\begin{itemize}
    \item There is only one walk (and path) of length 1;
    \item The number of walks grows exponentially with the number of nodes;
    \item The number of paths grows with the number of nodes, but at a rate inversely related to the number of nodes;
    \item The most frequent path length are the longest ones (lengths N-1 and N-2);
    \item It is always possible to get a walk longer than a previous one;
    \item There is no path of length longer than N-1.
\end{itemize}

Figure \ref{wpnC10} illustrates this result for the a complete graph with $10$ nodes (C-10).

\begin{figure}[!ht]
\centering\includegraphics[width=0.9\linewidth]{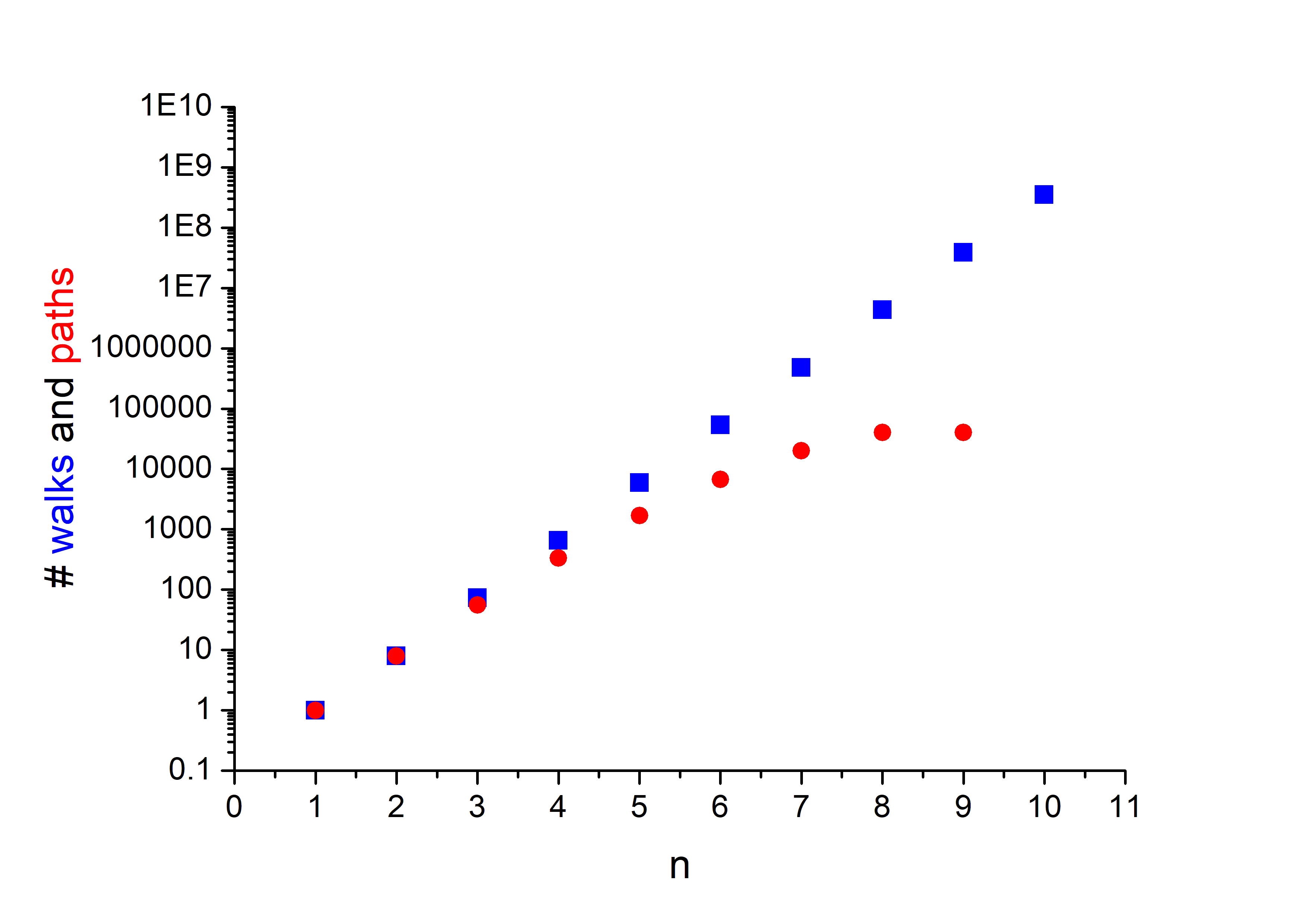}
\caption{Number of walks (blue squares) and paths (red circles) from 0 to 1, in the C10 graph, for each path length.}
\label{wpnC10}
\end{figure}

\subsection{Zachary's Karate Club}
\label{sec:zachary}

Zachary's Karate Club graph is a well known graph \cite{Zachary1977,Girvan2002}, with $N=34$, $M=78$, 1 connected component.
Figure \ref{fig:zachs} shows the Zachary's Karate Club graph, with numerated nodes.

\begin{figure}[H]
    \centering
    \includegraphics[scale=0.7]{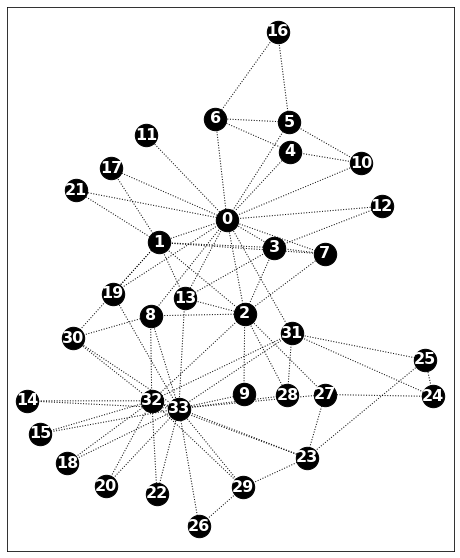}
    \caption{Zachary's Karate Club graph \cite{Zachary1977} has with $N=34$ vertices and $M=78$ edges and a single connected component.}
    \label{fig:zachs}
\end{figure}

Figure \ref{wpn} shows the distribution of the number of walks and paths with specific lengths from node 0 to node 1 in the Zachary's Karate Club graph.
Considering walks and paths between nodes $0$ and $1$, it is possible to note that:

\begin{itemize}
    \item There is an edge connecting nodes $0$ and $1$, so, $A_{01}=1$.
    \item As $(A^{1})_{01} = A_{01} = 1$, there is only $1$ walk from $0$ to $1$ with length $1$.
    \item That walk is a path as well.
    \item The shortest path length between nodes $0$ and $1$ is $1$: $\ppathS_{01} = 1$.
    \item As $(A^{2})_{01} = 7$, there are $7$ walks from $0$ to $1$ with length $2$.
    \item All those walks are paths as well.
    \item As $(A^{3})_{01} = 37$, there are $37$ walks from $0$ to $1$ with length $3$.
    \item However, only $13$ of those $37$ walks are paths as well.
    \item The length of the longest path between nodes 0 and 1 is 18.
\end{itemize}

The number of paths is calculated using the Depth-limited search (DLS), setting the path length as the limit of the DLS. The longer the path length, the more significant the difference between the number of walks and paths between a pair of nodes. There are infinite walks between nodes 0 and 1, but precisely 8.854.467.719.776.520.000 ($\approx$ 8E18) walks with lengths up to 18. The number of paths between nodes 0 and 1 is 80.137 ($\approx$ 8E4).

\begin{figure}[!h]
\centering\includegraphics[width=0.9\linewidth]{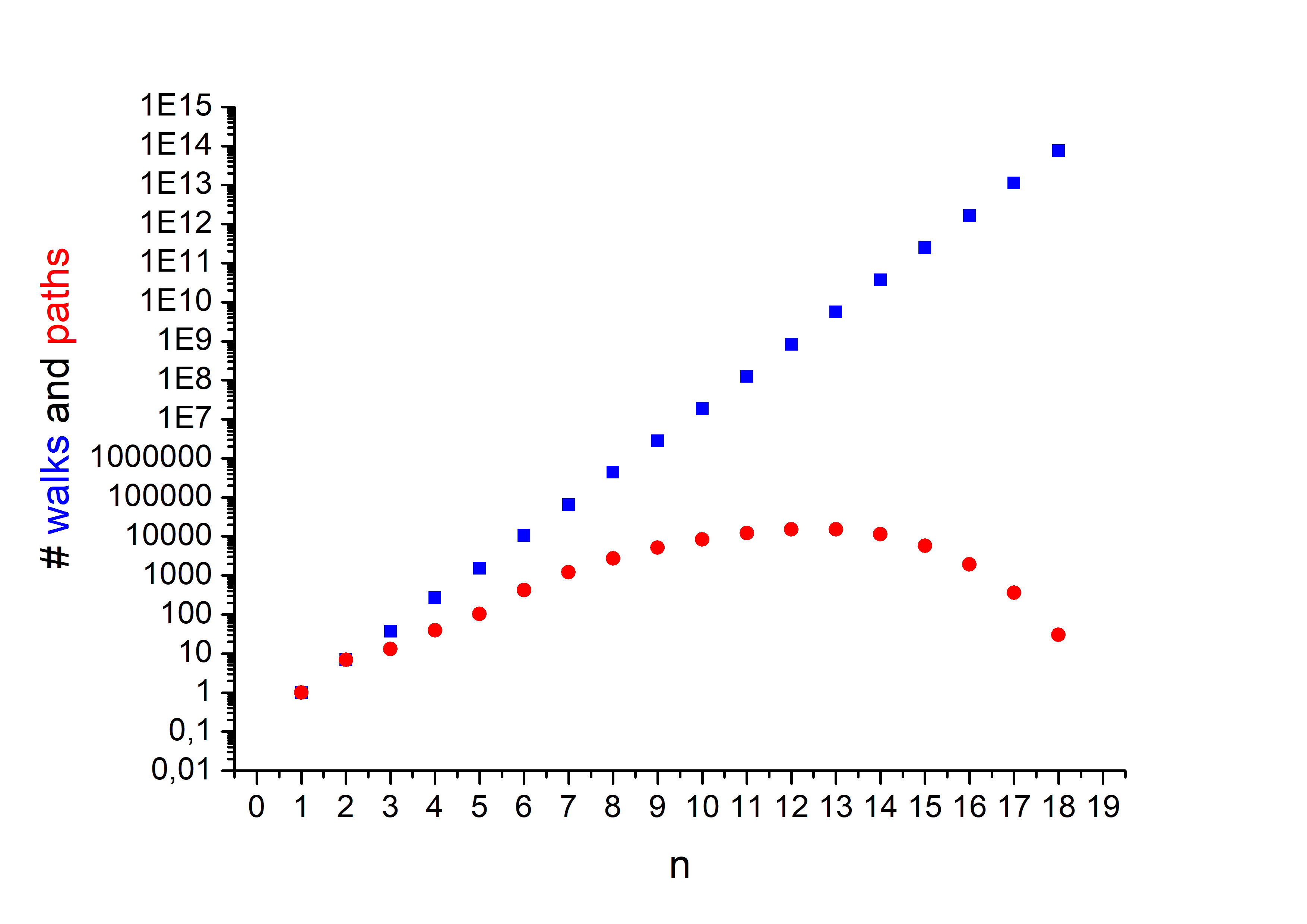}
\caption{Number of walks (blue squares) and paths (red circles) from node 0 to node 1, in the Zachary Karate Club graph, for each path length.}
\label{wpn}
\end{figure}

Going beyond the shortest path, let us calculate the expected value for $\walkL$, under the k-order approximation:
	\begin{equation}
    \mathbb{E}[\walkL_{ij}] =\frac
    {\sum_{n=\walkS}^{\walkS+k} n (A^n)_{ij}}
    {\sum_{n=\walkS}^{\walkS+k} (A^n)_{ij}}.
    \end{equation}

The expected value for $\walk_{01}$, under the k=17-order approximation is, thus:
	\begin{equation}
    \mathbb{E}[\walkL_{01}] =\frac
    {\sum_{n=1}^{18} n (A^n)_{01}}
    {\sum_{n=1}^{18} (A^n)_{01}}.
    \end{equation}

On the other hand, the expected value for $\ppathL$, under the k-order approximation is:
	\begin{equation}
		\mathbb{E}[\ppathL_{ij}] = \frac
    {\sum_{n=1}^{\ppathS_{ij}+k} n f_{P}^{(ij)}(n)}
    {\sum_{n=1}^{\ppathS_{ij}+k} f_{P}^{(ij)}(n)}
	\end{equation}

So, the expected value for $\ppathL_{01}$, under the k=17-order approximation:
	\begin{equation}
		\mathbb{E}[\ppathL_{01}] = \frac
    {\sum_{n=1}^{18} n f_{P}^{(ij)}(n)}
    {\sum_{n=1}^{18} f_{P}^{(ij)}(n)}
	\end{equation}

Figure \ref{delta} shows our ``delta measure'': $\mathbb{E}[\walk_{01}]-s\walk_{01}$ (for walks), $\mathbb{E}[\ppathL_{01}]-s\ppathL_{01}$ (for paths), for each value of $k$.

\begin{figure}[!h]
\centering\includegraphics[width=0.9\linewidth]{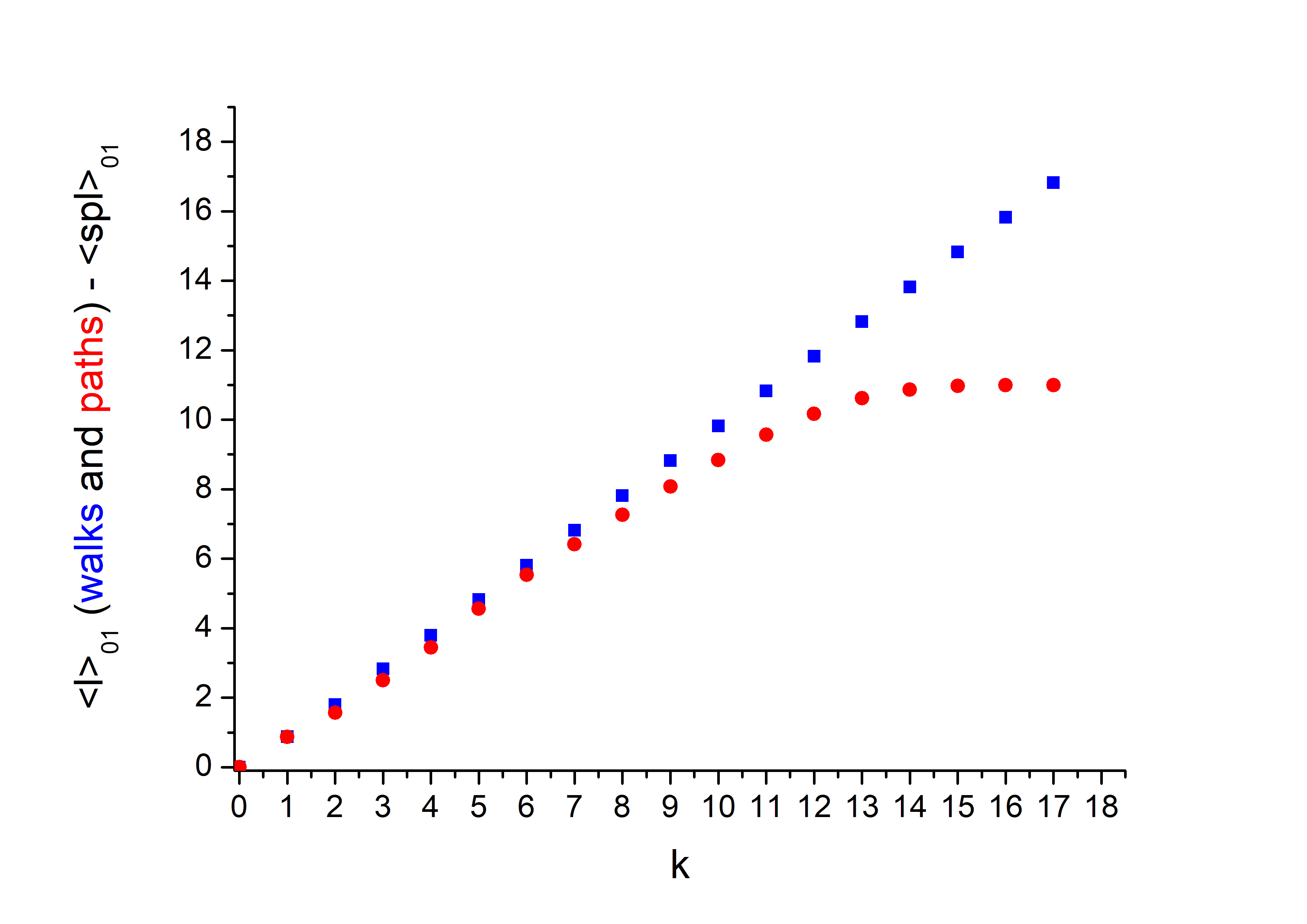}
\caption{$\mathbb{E}[\walk_{01}]-s\walk_{01}$ (blue squares) and $\mathbb{E}[\ppathL_{01}]-s\ppathL_{01}$ (red circles), in the Zachary karate club graph, for different k-order approximations}
\label{delta}
\end{figure}

It is essential to highlight that difference, for the walks-case, grows indefinitely.
However, in the case of the paths, it converges: it happens because when we allow longer paths, although the numerator increases (path length), that increment decreases - once, in a non-complete graph, the number of paths with a length close to the longest possible is smaller than the number of paths with intermediate lengths.

Considering all origins and destinations (all the nodes), we notice that the largest number of paths are between nodes 16 to 25, with 4319868 paths, going from length 4 to 23.
This is interesting since node 16 appears in every case where the shortest path length is equal to the diameter of the network, is on the top of the mean length of paths, and has the highest mode.
Node 16 seems the most unapproachable from the network by looking at these data. To access the table containing the statistics of all paths click \href{https://github.com/gioguarnieri/Pesquisa_Doutorado/blob/master/all_paths_data.md}{here}\footnote{\url{https://github.com/gioguarnieri/Pesquisa_Doutorado/blob/master/all_paths_data.md}}.

\section{Conclusions}
\label{sec:conclusions}

This short report presents an original idea about going ``beyond the shortest path''.
After presenting some fundamental concepts in graph theory, we presented an analytical solution for the problem of counting the number of possible paths between two nodes in complete graphs, and a depth-limited  approach to get all possible paths between each pair of nodes in a general graph.
Using the simple and well-known Zachary's karate club graph, we showed the distribution of walks and path lengths.

The most important result is that we can go beyond the shortest path (facing an NP-hard problem), but (fortunately) we do not need to go so far: there is a convergence/saturation value for the path-length expected value - once, in a non-complete graph, the number of paths with a length close to the longest possible is smaller than the number of paths with intermediate lengths. The value of that control parameter ($k$ - number of edges beyond the shortest path length) can be even smaller when considering penalties for longer paths.

In future work we plan to apply those ideas to a real-world problem, addressing urban mobility-related problems.


\bibliographystyle{elsarticle-num-names}
\bibliography{beyond.bib}

\begin{thebibliography}{10}
\expandafter\ifx\csname natexlab\endcsname\relax\def\natexlab#1{#1}\fi
\providecommand{\url}[1]{\texttt{#1}}
\providecommand{\href}[2]{#2}
\providecommand{\path}[1]{#1}
\providecommand{\DOIprefix}{doi:}
\providecommand{\ArXivprefix}{arXiv:}
\providecommand{\URLprefix}{URL: }
\providecommand{\Pubmedprefix}{pmid:}
\providecommand{\doi}[1]{\href{http://dx.doi.org/#1}{\path{#1}}}
\providecommand{\Pubmed}[1]{\href{pmid:#1}{\path{#1}}}
\providecommand{\bibinfo}[2]{#2}
\ifx\xfnm\relax \def\xfnm[#1]{\unskip,\space#1}\fi
\bibitem[{Barabási and Pósfai(2016)}]{barabasi2016network}
\bibinfo{author}{A.-L. Barabási}, \bibinfo{author}{M.~Pósfai},
  \bibinfo{title}{Network science}, \bibinfo{publisher}{Cambridge University
  Press}, \bibinfo{address}{Cambridge}, \bibinfo{year}{2016}. \URLprefix
  \url{http://barabasi.com/networksciencebook/}.
\bibitem[{Lima et~al.(2016)Lima, Stanojevic, Papagiannaki, Rodriguez, and
  González}]{Lima2016}
\bibinfo{author}{A.~Lima}, \bibinfo{author}{R.~Stanojevic},
  \bibinfo{author}{D.~Papagiannaki}, \bibinfo{author}{P.~Rodriguez},
  \bibinfo{author}{M.~C. González},
\newblock \bibinfo{title}{Understanding individual routing behaviour},
\newblock \bibinfo{journal}{J. R. Soc. Interface} \bibinfo{volume}{13}
  (\bibinfo{year}{2016}) \bibinfo{pages}{20160021}. \URLprefix
  \url{https://royalsocietypublishing.org/doi/10.1098/rsif.2016.0021}.
\bibitem[{Galbrun et~al.(2016)Galbrun, Pelechrinis, and Terzi}]{Galbrun2016}
\bibinfo{author}{E.~Galbrun}, \bibinfo{author}{K.~Pelechrinis},
  \bibinfo{author}{E.~Terzi},
\newblock \bibinfo{title}{Urban navigation beyond shortest route: The case of
  safe paths},
\newblock \bibinfo{journal}{Information Systems} \bibinfo{volume}{57}
  (\bibinfo{year}{2016}) \bibinfo{pages}{160--171}. \URLprefix
  \url{https://www.sciencedirect.com/science/article/pii/S0306437915001854}.
  \DOIprefix\doi{https://doi.org/10.1016/j.is.2015.10.005}.
\bibitem[{Tomas et~al.(2022)Tomas, Soares, Jorge, Mendes, Freitas, and
  Santos}]{Tomas2022}
\bibinfo{author}{L.~R. Tomas}, \bibinfo{author}{G.~G. Soares},
  \bibinfo{author}{A.~A.~S. Jorge}, \bibinfo{author}{J.~F. Mendes},
  \bibinfo{author}{V.~L.~S. Freitas}, \bibinfo{author}{L.~B.~L. Santos},
\newblock \bibinfo{title}{Flood risk map from hydrological and mobility data: a
  case study in são paulo (brazil)},
\newblock \bibinfo{journal}{Transacitons in GIS - accepted for publication}
  (\bibinfo{year}{2022}). \URLprefix
  \url{https://onlinelibrary.wiley.com/doi/10.1111/tgis.12962}.
\bibitem[{Estrada and Hatano(2008)}]{Estrada2008}
\bibinfo{author}{E.~Estrada}, \bibinfo{author}{N.~Hatano},
\newblock \bibinfo{title}{Communicability in complex networks},
\newblock \bibinfo{journal}{Physical review. E, Statistical, nonlinear, and
  soft matter physics} \bibinfo{volume}{77} (\bibinfo{year}{2008})
  \bibinfo{pages}{036111}. \URLprefix \url{https://arxiv.org/abs/0707.0756}.
\bibitem[{Biggs(1974)}]{Biggs}
\bibinfo{author}{N.~Biggs}, \bibinfo{title}{Algebraic graph theory},
  \bibinfo{publisher}{Cambridge university press}, \bibinfo{year}{1974}.
\bibitem[{Roberts and Kroese(2007)}]{Roberts2007}
\bibinfo{author}{B.~Roberts}, \bibinfo{author}{D.~P. Kroese},
\newblock \bibinfo{title}{Estimating the number of s-t paths in a graph},
\newblock \bibinfo{journal}{Journal of Graph Algorithms and Applications}
  \bibinfo{volume}{11} (\bibinfo{year}{2007}) \bibinfo{pages}{195–214}.
  \URLprefix \url{https://jgaa.info/accepted/2007/RobertsKroese2007.11.1.pdf}.
\bibitem[{Zachary(1977)}]{Zachary1977}
\bibinfo{author}{W.~W. Zachary},
\newblock \bibinfo{title}{An information flow model for conflict and fission in
  small groups},
\newblock \bibinfo{journal}{Journal of Anthropological Research}
  \bibinfo{volume}{33} (\bibinfo{year}{1977}) \bibinfo{pages}{452--473}.
  \URLprefix \url{https://www.jstor.org/stable/3629752}.
\bibitem[{Herrero(2005)}]{SAW}
\bibinfo{author}{C.~P. Herrero},
\newblock \bibinfo{title}{Self-avoiding walks on scale-free networks},
\newblock \bibinfo{journal}{Phys. Rev. E} \bibinfo{volume}{71}
  (\bibinfo{year}{2005}) \bibinfo{pages}{016103}. \URLprefix
  \url{https://link.aps.org/doi/10.1103/PhysRevE.71.016103}.
  \DOIprefix\doi{10.1103/PhysRevE.71.016103}.
\bibitem[{Girvan and Newman(2002)}]{Girvan2002}
\bibinfo{author}{M.~Girvan}, \bibinfo{author}{M.~E. Newman},
\newblock \bibinfo{title}{Community structure in social and biological
  networks},
\newblock \bibinfo{journal}{Proceedings of the national academy of sciences}
  \bibinfo{volume}{99} (\bibinfo{year}{2002}) \bibinfo{pages}{7821--7826}.

\end{thebibliography}


    



\end{document}